\documentclass{article}
\usepackage[english]{babel}
\usepackage{amsmath,amssymb,amsthm}

\usepackage[pdfpagemode=UseOutlines,
bookmarks=true, bookmarksopenlevel=3, bookmarksopen=true,
bookmarksnumbered=true, unicode=true,
pdffitwindow=true,pdfpagelayout=SinglePage,pdfhighlight=/P,
colorlinks=true,linkcolor=blue,citecolor=blue,urlcolor=blue]{hyperref}

\advance\textwidth 20mm  \advance\oddsidemargin -10mm
\advance\textheight 22mm \advance\topmargin -15mm

\allowdisplaybreaks[4]

\theoremstyle{plain}
\newtheorem{proposition}{Proposition}

\theoremstyle{remark}

\title{Non-autonomous reductions of the KdV equation and multi-component analogs of the Painlev\'e equations P$_{34}$ and P$_3$}

\author{V.E.\:Adler\thanks{L.D.\:Landau Institute for Theoretical Physics, Akademika Semenova av. 1A, 142432 Chernogolovka, Russian Federation (permanent address). E-mail: adler@itp.ac.ru} $^\ddag$,
M.P.\:Kolesnikov\thanks{Moscow Institute of Physics and Technology, Institutskiy 9, 141701 Dolgoprudny, Russian Federation (permanent address).} \thanks{Institute of Mathematics, Ufa Federal Research Centre, Russian Academy of Sciences, Chernyshevsky str. 112, 450008 Ufa, Russian Federation.}}

\date{19 April 2023}

\begin{document}
\maketitle

\begin{abstract}
We study reductions of the Korteweg--de Vries equation corresponding to stationary equations for symmetries from the noncommutative subalgebra. An equivalent system of $n$ second-order equations is obtained, which reduces to the Painlev\'e equation P$_{34}$ for $n=1$. On the singular line $t=0$, a subclass of special solutions is described by a system of $n-1$ second-order equations, equivalent to the P$_3$ equation for $n=2$. For these systems, we obtain the isomonodromic Lax pairs and B\"acklund transformations which form the group ${\mathbb Z}^n_2\times{\mathbb Z}^n$.
\end{abstract}

\section{Introduction}

It is well known that the Lie algebra of generalized symmetries of the KdV equation
\[
 u_t=u_{xxx}-6uu_x
\]
is generated by two infinite sequences of flows (see e.g. \cite{Ibragimov_Shabat_1979, Fuchssteiner_1983, Orlov_Shulman_1985, Burtsev_Zakharov_Mikhailov_1987})
\[
 u_{t_n}=R^n(u_x),\quad 2u_{\tau_n}=R^n(6tu_x-1)=6tu_{t_n}-R^n(1),\quad R=D^2_x-4u-2u_xD^{-1}_x,
\]
where $R$ is called the recursion operator. The first sequence is the usual higher symmetries, commutative and local. The second sequence of symmetries generates a noncommutative Lie subalgebra. These flows are non-local, with the exception of the first two, corresponding to the classical Galilean and scaling symmetries 
\begin{equation}\label{tau01}
 2u_{\tau_0}=6tu_x-1,\quad u_{\tau_1}=3tu_t+xu_x+2u,
\end{equation}
for example, the next flow (KdV master symmetry) is of the form
\begin{equation}\label{tau2}
 u_{\tau_2}=3t(u_{xxxx}-10uu_{xx}-5u^2_x+10u^3)_x+xu_t+4u_{xx}-8u^2-2u_xD^{-1}(u).
\end{equation}
Stationary equations for symmetries of both types determine constraints compatible with KdV, of the general form
\begin{equation}\label{ab}
 a[R](u_x)+b[R](6tu_x-1)=0
\end{equation}
where $a$ and $b$ are operator polynomials with constant coefficients. If $b[R]=0$, then this is the Novikov equation that defines finite-gap solutions, has a complete set of first integrals, and is Liouville integrable. It was noticed by Choodnovsky and Choodnovsky \cite{Choodnovsky_1978} that the Novikov equation can be transformed to the Garnier system \cite{Garnier_1919} for the eigenfunctions of the Schr\"odinger operator:
\begin{equation}\label{Garnier}
 u=\psi_1\varphi_1+\dots+\psi_n\varphi_n,\quad 
 \psi_{j,xx}=(u-\lambda_j)\psi_j,\quad \varphi_{j,xx}=(u-\lambda_j)\varphi_j.
\end{equation}
This can also be interpreted as the stationary equation for a sum of the simplest symmetry $u_{t_0}=u_x$ and the negative KdV flows which are defined as squared eigenfunction symmetries $u_{t_j}= (\psi_j\varphi_j)_x$. This topic was developed and applied to other equations by Antonowicz and Rauch-Wojciechowski \cite{Antonowicz_Rauch-Wojciechowski_1990, Antonowicz_Rauch-Wojciechowski_1991, Antonowicz_1992} and others, see also the book \cite{Suris_2003} containing comprehensive information about the Garnier system and a detailed bibliography.

For $b[R]\ne0$, equation (\ref{ab}) is not Liouville integrable, but retains the Painlev\'e property. These are the so-called string equations \cite{Moore_1990, Novikov_1990}. Reductions of this type are of interest as a potential source of exact solutions with nonstandard asymptotics. In particular, they have applications to the Gurevich--Pitaevskii problems on the decay of the wave front in the vicinity of the breaking point and on the evolution of step-like potentials \cite{GP_1973a, GP_1973b}; in this connection, we mention the results by Suleimanov and others \cite{Suleimanov_1994, Kudashev_Suleimanov_1996, Dubrovin_2006, Adler_2020}. 

If the polynomial $b[R]$ is linear, that is, only classical symmetries (\ref{tau01}) are added to the Novikov equation, then equation (\ref{ab}) turns into the hierarchies of higher P$_1$ and P$_2$ Painlev\'e equations \cite{Kudryashov_1997, Kudryashov_2002, Clarkson_Joshi_Pickering_1999, Mazzocco_Mo_2007}. The stationary equation for a linear combination of flows (\ref{tau01}) and (\ref{tau2}) was studied in \cite{Adler_2020}.

In this paper, we consider the case when $a[R]=0$ and $b[R]$ is a polynomial of arbitrary degree $n$ without multiple roots. The main result is that, despite the nonlocal nature of the symmetries used, the stationary equation for them can be written in a compact and uniform form.

In Section \ref{s:stflow}, we show that the equation $b[R]=0$ can be transformed into a stationary equation for the Galilean symmetry and a sum of the negative flows of the KdV hierarchy, which leads to the system (\ref{yxx}), ( \ref{uy}) of $n$ nonautonomous second-order equations. The evolution in $t$ in virtue of the KdV equation is described by additional system (\ref{yt}). For $n=1$, the Painlev\'e equation P$_{34}$ appears, which defines the usual self-similar KdV reduction. Note that, for $n=2$, close Painlev\'e type systems arising from the self-similar reduction in the 5th order Sawada--Kotera and Kaup--Kuperschmidt equations were studied by Hone \cite{Hone_2001}.

Interpretation of negative flows as symmetries with squared eigenfunctions leads to a non-autonomous version of the Garnier system, which differs from (\ref{Garnier}) only in the form of relation between the potential and psi-functions:
\[
 3tu-\frac{x}{2}=\psi_1\varphi_1+\dots+\psi_n\varphi_n.
\]
If we choose in (\ref{ab}) a polynomial $a[R]$ of degree greater than that of $b[R]$, then more general constraints arise, but their study is beyond the scope of our work. A general construction of non-autonomous generalizations of the Garnier system, in which symmetries with squared eigenfunctions are immediately added into the constraint (\ref{ab}), was proposed by Orlov and Rauch-Wojciechowski \cite{Orlov_Rauch-Wojciechowski_1993}. 

In Section \ref{s:DBT}, B\"acklund transformations for systems (\ref{yxx})--(\ref{uy}) are derived, based on the Darboux transformations for the auxiliary linear problems, and it is shown that they form the group ${\mathbb Z}^n_2\times{\mathbb Z}^n$. 

For the system (\ref{yt}), the parameter value $t=0$ is a singular point and the corresponding KdV solution has, in general, a singularity along the line $t=0$. In section \ref{s:t0}, a subclass of special regular solutions is distinguished. For them, the order of the system (\ref{yxx}) is reduced for $t=0$, so that the initial conditions on this line are described by a system of $n-1$ second-order equations, which is equivalent to the P$_3$ equation for $n=2$. The group of B\"acklund transformations for this system remains the same.

\section{Main equations}\label{s:stflow} 

We study the following two systems of non-autonomous ODEs with respect to the variables $y_1,y_{1,x}$, \dots, $y_n,y_{n,x}$:
\begin{gather}
\label{yxx}
 y_{j,xx}=\frac{y^2_{j,x}-\alpha^2_j}{2y_j}+2(u-\lambda_j)y_j,\\
\label{yt}
 y_{j,t}=2u_xy_j-2(u+2\lambda_j)y_{j,x},\\
\label{uy}
 u:= \frac{x}{6t}+\frac{1}{3t}(y_1+\dots+y_n)
\end{gather}
where $\alpha_1,\dots,\alpha_n$ and $\lambda_1,\dots,\lambda_n$ are arbitrary parameters such that $\lambda_i\ne\lambda_j$ for $i\ne j$, and it is understood that extension of equations (\ref{yt}) to derivatives $y_{j,x}$ is determined by differentiation with respect to $x$ by virtue of (\ref{yxx}).

For $n=1$, the system (\ref{yxx}), (\ref{uy}) reduces to one equation
\[
 y_{1,xx}=\frac{y^2_{1,x}-\alpha^2_1}{2y_1}+2\Bigl(\frac{y_1}{3t}+\frac{x}{6t}-\lambda_1\Bigr)y_1.
\]
Up to obvious substitutions, it coincides with the Painlev\'e equation P$_{34}$ 
\[
 w_{zz}=\frac{w^2_z-\alpha^2}{2w}+2w^2-zw
\]
describing the self-similar reduction of the KdV equation. We will show that the equations for arbitrary $n$ also define a certain KdV reduction and derive an isomonodromy representation and B\"acklund transformations for them.

Note that equations (\ref{yxx})--(\ref{uy}) are related by the substitution $y_j=\psi_j\varphi_j$ with the Garnier type systems
\begin{equation}\label{G}
\begin{gathered}
 \psi_{j,xx}=(u-\lambda_j)\psi_j,\quad
 \varphi_{j,xx}=(u-\lambda_j)\varphi_j,\\
 \psi_{j,t}=u_x\psi_j-2(u+2\lambda_j)\psi_{j,x},\quad
 \varphi_{j,t}=u_x\varphi_j-2(u+2\lambda_j)\varphi_{j,x},\\
 u:= \frac{x}{6t}+\frac{1}{3t}(\psi_1\varphi_1+\dots+\psi_n\varphi_n);
\end{gathered}
\end{equation}
the parameters $\alpha_j=\psi_{j,x}\varphi_j-\psi_j\varphi_{j,x}$ are the first integrals for (\ref{G}).

The main property of the systems (\ref{yxx})--(\ref{uy}) is formulated as follows.

\begin{proposition}\label{prop:xt}
Equations (\ref{yxx})--(\ref{uy}) are consistent, that is, the equalities $(y_{j,xx})_t=(y_{j,t})_{xx}$ hold identically. By virtue of (\ref{yxx})--(\ref{uy}), the variable $u$ satisfies the KdV equation
\begin{equation}\label{ut}
 u_t=u_{xxx}-6uu_x.
\end{equation}
\end{proposition}

We present a proof based on the use of auxiliary linear problems for non-autonomous symmetries:
\[
 \Psi_x=U\Psi,\quad \Psi_t=V\Psi,\quad \Psi_\tau+\kappa\Psi_\lambda=W\Psi
\]
where $U,V,W$ are matrices depending on $x,t,\tau$ and polynomial in the spectral parameter $\lambda$, and $\kappa$ is a
polynomial in $\lambda$ with constant coefficients. The compatibility conditions are written in the form of the zero curvature
representations 
\begin{equation}\label{UVW}
 U_t=V_x+[V,U],\quad U_\tau+\kappa U_\lambda=W_x+[W,U],\quad V_\tau+\kappa V_\lambda=W_t+[W,V].
\end{equation}
For the KdV hierarchy, the matrices $U,V$ and $W$ have a common structure
\begin{equation}\label{M}
 M(g)=\begin{pmatrix}
  -g_x & 2g \\
  2(u-\lambda)g-g_{xx} & g_x
 \end{pmatrix}.
\end{equation}
Let
\begin{equation}\label{UV}
 U=M(1/2)=\begin{pmatrix}
  0 & 1 \\
  u-\lambda & 0
 \end{pmatrix},\quad 
 V=M(-u-2\lambda),
\end{equation}
then the first equation (\ref{UVW}) is equivalent to the KdV equation (\ref{ut}) and the varables $\psi_j$ and $\varphi_j$ from (\ref{G}) are interpreted as solutions to the corresponding linear problems for $\lambda=\lambda_j$. Given the above $U$, $V$ and a general matrix $W=M(g)$, the last two equations (\ref{UVW}) reduce to the relations
\begin{gather}
\label{utau}
 u_\tau=-g_{xxx}+4(u-\lambda)g_x+2u_xg+\kappa,\\
\label{gt}
 g_t=g_{xxx}-6ug_x-3\kappa.
\end{gather}
The equations (\ref{ut}), (\ref{utau}), and (\ref{gt}) satisfy the compatibility conditions $(u_t)_\tau=(u_\tau)_t$, which is easy to check directly (in other words, the equations (\ref{utau}), (\ref{gt}) define a generalized symmetry for the KdV equation). Hence it follows that the stationary equation
\begin{equation}\label{gxxx}
 g_{xxx}+4(\lambda-u)g_x-2u_xg=\kappa
\end{equation}
is a constraint consistent with (\ref{gt}) and (\ref{ut}). Taking this into account, the Proposition \ref{prop:xt} is a simple corollary of the following statement.

\begin{proposition}\label{prop:yg}
Assume that $\deg g=\deg\kappa=n$ and the zeroes of the polynomial $\kappa$ are simple, then equations (\ref{yxx})--(\ref{yt}) are equivalent to (\ref{gxxx}) and (\ref{gt}).
\end{proposition}
\begin{proof} Let
\[
 g(\lambda)=g_0\lambda^n+\dots+g_n,\quad 
 \kappa(\lambda)=\kappa_0\lambda^n+\dots+\kappa_n
  =\kappa_0(\lambda-\lambda_1)\cdots(\lambda-\lambda_n)
\]
where $\kappa_0\ne0$ and all $\lambda_j$ are pairwise different. By using the constraint (\ref{gxxx}), we bring (\ref{gt}) to the equivalent form
\begin{equation}\label{gt'}
 g_t=2u_xg-2(u+2\lambda)g_x-2\kappa.
\end{equation}
Let us set $Y_j=g(\lambda_j)$. For $\lambda=\lambda_j$, equation (\ref{gxxx}) admits the integrating factor $2Y_j$ and equations (\ref{gxxx}), (\ref{gt'}) reduce to
\[
 2Y_jY_{j,xx}-Y^2_{j,x}+4(\lambda_j-u)Y^2_j=\gamma_j,\quad Y_{j,t}=2u_xY_j-2(u+2\lambda_j)Y_{j,x}.
\]
Checking the compatibility condition shows that the integration constants $\gamma_j$ do not depend on $t$. To close these equations, it is necessary to relate $u$ and $Y_j$, for which it suffices to express the first two coefficients $g_0$ and $g_1$ in terms of $u$. The equality of the coefficients at the powers of $\lambda$ in (\ref{gxxx}) gives the recurrent relations
\begin{equation}\label{gjx}
 g_{0,x}=0,~~ -4g_{j+1,x}=g_{j,xxx}-4ug_{j,x}-2u_xg_j-\kappa_j,~~ j=0,\dots,n-1,
\end{equation}
while equation (\ref{gt}) is equivalent to relations
\begin{equation}\label{gjt}
 g_{j,t}=g_{j,xxx}-6ug_{j,x}-3\kappa_j,\quad j=0,\dots,n,
\end{equation}
which make it possible to refine the dependence on $t$ of the integration constants arising in the calculation of $g_j$. For the first two coefficients it is easy to show that
\[
 g_0=-3\kappa_0t+c_0,\quad
 g_1=-\frac{\kappa_0}{4}(6tu-x)+\frac{c_0}{2}u-3\kappa_1t+c_1,\quad c_0,c_1=\operatorname{const}.
\]
Applying transformations of the form $t\to t-t_0$, $x\to x-x_0$, we set $c_0=c_1=0$ without loss of generality, then
\begin{equation}\label{gnn}
 g_0=-3\kappa_0t,\quad g_1=-\frac{\kappa_0}{4}(6tu-x)-3\kappa_1t.
\end{equation}
Hence it follows that $g+3\kappa t$ is polynomial in $\lambda$ of degree $n-1$. Its interpolation polynomial with nodes in $\lambda_1,\dots,\lambda_n$ is of the form
\[
 g+3\kappa t= -\frac{\kappa_0}{4}(6tu-x)\lambda^{n-1}+\dots
  = \sum^n_{j=1}Y_j\prod_{i\ne j}\frac{\lambda-\lambda_i}{\lambda_j-\lambda_i},
\]
which yields the expression for $u$ in terms of $Y_1,\dots,Y_n$:
\[
 \kappa_0(6tu-x)=-4\sum^n_{j=1}\frac{Y_j}{\prod_{i\ne j}(\lambda_j-\lambda_i)}.
\]
Now, in order to obtain equations (\ref{yxx})--(\ref{yt}), we only have to replace
\[
 y_j= -\frac{2Y_j}{\kappa_0\prod_{i\ne j}(\lambda_j-\lambda_i)}
\]  
and to change the constants $\gamma_j$. 
\end{proof}

It is easy to see that the variables $y_j$ are nothing but the coefficients of the pole expansion of $g/\kappa$:
\[
 \frac{g}{\kappa}=-3t-\frac{1}{2}\Bigl(\frac{y_1}{\lambda-\lambda_1}+\dots+\frac{y_n}{\lambda-\lambda_n}\Bigr).
\]
Passing to the matrix $\tilde W=W/\kappa=M(g/\kappa)$, we obtain isomonodromic Lax pairs for our systems.
\begin{proposition}\label{prop:xtLax}
Equations (\ref{yxx})--(\ref{uy}) admit the isomonodromy representations
\begin{equation}\label{WxWt}
 W_x=U_\lambda+[U,W],\quad W_t=V_\lambda+[V,W],
\end{equation}
with the matrices $U$, $V$ given in (\ref{UV}) and
$\displaystyle W=M\left(-3t-\frac{1}{2}\Bigl(\frac{y_1}{\lambda-\lambda_1}+\dots+\frac{y_n}{\lambda-\lambda_n}\Bigr)\right)$.
\end{proposition}

To conclude this section, let us comment on the connection between the equations under study and the stationary equations for KdV symmetries mentioned in the Introduction.

1) Under the choice $U=M(1/2)$ and $V=\frac{1}{4}(\mu-\lambda)^{-1}M(y)$, the zero curvature equation $U_\tau=V_x+[V,U]$ 
defines the so-called negative symmetry of KdV, that is, a non-local flow of the form
\[
 u_\tau=y_x,\quad y_{xxx}-4(u-\mu)y_x-2u_xy=0.
\]
The last equation can be integrated, resulting in an equation of the form (\ref{yxx}). Choosing different values of the parameter $\mu=\lambda_j$, we obtain a family of flows 
\[
 u_{\tau_j}=y_{j,x},\quad y_{j,xx}=\frac{y^2_{j,x}-\alpha^2_j}{2y_j}+2(u-\lambda_j)y_j,
\]
which are compatible with the KdV equation. Hence it follows that the constraint (\ref{uy}) is nothing but the integrated stationary equation for the sum of such flows with the classical Galilean symmetry:
\[
 u_\tau=6tu_x-1-2(y_{1,x}+\dots+y_{n,x})=0.
\]

2) Recurrent relations (\ref{gjx}) can be written as
\[
 -4g_{j+1,x}=R(g_{j,x})-\kappa_j,\quad R=D^2_x-4u-2u_xD^{-1}_x
\]
and one can prove by induction that the flow (\ref{utau}) is $u_\tau=-\kappa[-\frac{1}{4}R](6tu_x-1)$. Thus, the system (\ref{yxx}), (\ref{uy}) turns out to be equivalent to the stationary equation for a symmetry from the Lie subalgebra of nonlocal KdV symmetries, assuming that the zeros of $\kappa$ are simple.

3) From (\ref{gt}) it follows that $\deg g\ge\deg\kappa$; we have considered the case when both degrees are the same. If $n=\deg g>\deg\kappa$ then the recurrent relations (\ref{gjx}) remain valid, but first $r$ coefficients $\kappa_0,\dots,\kappa_{r-1}$ may vanish. In such a case, the coefficients $g_0,\dots,g_{r+1}$ are local and define the autonomous symmetries from the main Lie subalgebra. The resulting stationary equation is of the form
\[
\left\{\begin{gathered}
 y_{j,xx}=\frac{y^2_{j,x}-\alpha^2_j}{2y_j}+2(u-\lambda_j)y_j,\quad j=1,\dots,n-r,\\
 \sum^r_{j=1}a_jQ^j(u)+6tu-x=y_1+\dots+y_{n-r}
\end{gathered}\right.
\]
where $Q=D^{-1}_xRD_x=D^2_x-4u+2D^{-1}_xu_x$ (the corresponding systems of Garnier type were proposed in \cite{Orlov_Rauch-Wojciechowski_1993}). For instance, the case $r=1$ corresponds to the constraint
\[
 a_1(u_{xx}-3u^2)+6tu-x=y_1+\dots+y_{n-r},
\]
$r=2$ corresponds to the constraint
\[
 a_2(u_{xxxx}-10uu_{xx}-5u^2_x+10u^3)+a_1(u_{xx}-3u^2)+6tu-x=y_1+\dots+y_{n-2},
\]
and so on, until the case $r=n$ which corresponds just to a sum of local higher symmetries and the Galilean symmetry (the hierarchy of P$_1$ equation).

\section{Darboux--B\"acklund transformations}\label{s:DBT} 

Equations (\ref{yxx})--(\ref{uy}) admit B\"acklund transformations $A_1,B_1,\dots,A_n,B_n$, each of which changes one of the parameters $\alpha_j$ (the parameters $\lambda_j$ will never change; in particular, this fixes the obvious $S_n$ symmetry of our systems). The transformation $A_k$ acts on solutions trivially and only changes the sign of $\alpha_k$, without actual changing the equation:
\begin{equation}\label{y.Ak}
 A_k:\quad
 \left\{\begin{aligned} 
  &\tilde y_j=y_j,\\
  &\tilde\alpha_j=\alpha_j,~~j\ne k,\quad \tilde\alpha_k=-\alpha_k.
 \end{aligned}\right.
\end{equation}
The transformation $B_k$ is constructed with the help of auxiliary function
\begin{equation}\label{fk}
 f_k=\frac{y_{k,x}+\alpha_k}{2y_k},
\end{equation}
according to the formulae
\begin{equation}\label{y.Bk}
 B_k:\quad
  \left\{\begin{aligned}
   & \tilde y_j=\frac{(y_{j,x}-2y_jf_k)^2-\alpha^2_j}{4(\lambda_j-\lambda_k)y_j},\quad 
     \tilde\alpha_j=\alpha_j,\quad j\ne k,\\
   & \tilde y_k=-y_k+6t(f^2_k+\lambda_k)-x-\sum_{j\ne k}(\tilde y_j+y_j),\quad 
     \tilde\alpha_k=1-\alpha_k.
  \end{aligned}\right.
\end{equation}

\begin{proposition}\label{prop:Bk}
The transformations (\ref{fk}), (\ref{y.Bk}) preserve the form of the systems (\ref{yxx})--(\ref{uy}). Transformations $B_k$ are involutive and permutable with each other, as well as with transformations $A_j$ for $j\ne k$:
\begin{equation}\label{ABgroup}
 A^2_k=\operatorname{id},~~ A_jA_k=A_kA_j,~~ 
 B^2_k=\operatorname{id},~~ B_jB_k=B_kB_j,~~ A_jB_k=B_kA_j,~~ j\ne k.
\end{equation}
The group generated by $A_1,B_1,\dots,A_n,B_n$ is isomorphic to ${\mathbb Z}^n_2\times{\mathbb Z}^n$.
\end{proposition}

\begin{proof}
The $B_k$ transformations are derived from the Darboux transformations for equations (\ref{WxWt}) (see e.g.~\cite{Veselov_Shabat_1993}). They are defined by the formula $\tilde\Psi=F\Psi$ where the matrix $F$ satisfies the compatibility conditions
\begin{equation}\label{FUVW}
 F_x=\tilde UF-FU,\quad F_t=\tilde VF-FV,\quad F_\lambda=\tilde WF-FW.
\end{equation}
For the matrices $U$ and $V$ of the form (\ref{UV}) corresponding to the KdV equation, the matrix $F$ reads
\[
 F=(\lambda-\mu)^{-1/2}\begin{pmatrix}
   -f & 1 \\
   f^2+\mu-\lambda & -f
 \end{pmatrix}.
\]
The first two equations (\ref{FUVW}) give the standard formulae for the B\"acklund transformation defined as a composition of two Miura transformations between KdV and mKdV:
\begin{equation}\label{uf}
 u=f_x+f^2+\mu,\quad \tilde u=u-2f_x,\quad f_t=f_{xxx}-6(f^2+\mu)f_x.
\end{equation}
Since $\operatorname{tr}W=0$, the third equation (\ref{FUVW}) implies the condition $(\det F)_\lambda=0$ which explains why the factor $(\lambda-\mu)^{-1/2}$ was introduced into $F$. Let $W=M(y)$, then the remaining consequences of this equation are equivalent to the relations
\begin{equation}\label{yf}
 y_{xx}-2(fy)_x=2(\lambda-\mu)(\tilde y-y),\quad
 \tilde y_x+y_x+2f(\tilde y-y)=\frac{1}{2(\lambda-\mu)}.
\end{equation}
Recall that for our systems $y$ is of the form $y=-3t-\frac{1}{2}\sum_j\frac{y_j}{\lambda-\lambda_j}$ where all $\lambda_j$ are different. Then it follows from the second equation that the parameter $\mu$ must coincide with one of the poles. Let $\mu=\lambda_k$ for some fixed index $k$, then (\ref{yf}) reduces to the equations
\begin{equation}\label{yyf}
\begin{gathered}
 y_{j,xx}=2(fy_j)_x+2(\lambda_j-\lambda_k)(\tilde y_j-y_j),\quad j=1,\dots,n,\\
 -6tf_x=\tilde y_1+\dots+\tilde y_n-y_1-\dots-y_n,\\
 (\tilde y_j+y_j)_x+2f(\tilde y_j-y_j)=\left\{\begin{array}{cl} 0, & j\ne k,\\ -1, & j=k.\end{array}\right.  
\end{gathered}
\end{equation}
Integrating the first equality for $j=k$, we obtain $y_{k,x}=2fy_k+\beta$. Moreover, from (\ref{yxx}) and (\ref{uf}) (with $\mu=\lambda_k$) it follows that
\[
 y_{k,xx}=\frac{y^2_{k,x}-\alpha^2_k}{2y_k}+2(f_x+f^2)y_k,
\]
and these two relations imply $\beta=\pm\alpha_k$. We choose one of two values and define $f=(y_{k,x}-\beta)/(2y_k)$, then the first two equations (\ref{yyf}) give formulae for $\tilde y_1,\dots,\tilde y_n$. After that, one can check by direct calculations that the third equation (\ref{yyf}) is fulfilled identically, and also determine the values of the $\tilde\alpha_j$ parameters in the equations for $\tilde y_j$, which gives the formulae (\ref{y.Bk}).

The group identities for $A_j$ are verified directly, for $B_j$ they follow from the commutativity of the Darboux transformations for the Schr\"odinger operator. It is easy to see that the transformation $B_kA_k$ shifts the parameter $\alpha_k$ by 1 and has an infinite order. Since transformations with different numbers commute, all B\"acklund transformations generate the lattice of parameters $(\varepsilon_1\alpha_1+m_1,\dots,\varepsilon_n\alpha_n+m_n)$ where $\varepsilon_j=\pm1$, $m_j\in{\mathbb Z}$.
\end{proof}

Equation (\ref{fk}) suggests the change of variables from $y_1,y_{1,x},\dots,y_n,y_{n,x}$ to $y_1,f_1,\dots,y_n,f_n$, which brings equations (\ref{yxx})--(\ref{uy}) to a simpler polynomial form:
\begin{gather}
\label{yfx}
 y_{j,x}=2y_jf_j-\alpha_j,\quad f_{j,x}=u-f^2_j-\lambda_j,\\
\label{yft}
 y_{j,t}=2u_xy_j-2(u+2\lambda_j)y_{j,x},\quad f_{j,t}=(u_x-2(u+2\lambda_j)f_j)_x,\\
\label{uy'}
 u:= \frac{x}{6t}+\frac{1}{3t}(y_1+\dots+y_n)
\end{gather}
where $j=1,\dots,n$ and it is assumed, as before, that the $x$-derivative in (\ref{yft}) is defined by equations (\ref{yfx}) and (\ref{uy'}). 

Note that the variables $f_j$ satisfy, by virtue of these systems, the modified mKdV equations
\begin{equation}\label{mKdV}
 f_{j,t}=f_{j,xxx}-6(f^2_j+\lambda_j)f_{j,x}
\end{equation}
and the variables $y_j$ satisfy the degenerate Calogero-Degasperis equations (in the rational form)
\begin{equation}\label{CD}
 y_{j,t}=y_{j,xxx}-\frac{3y_{j,x}y_{j,xx}}{y_j}+\frac{3y_{j,x}(y^2_{j,x}-\alpha^2_j)}{2y^2_j}-6\lambda_jy_{j,x}.
\end{equation}
The first equation (\ref{yfx}) defines a differential substitution from (\ref{CD}) to (\ref{mKdV}) and the second equation (\ref{yfx}) is the Miura map from (\ref{mKdV}) to the KdV equation (\ref{ut}).

It can be verified by direct calculation that the B\"acklund transformations are written as birational mappings in the variables $y_j$ and $f_j$:
\begin{gather}
\label{yf.Ak}
 A_k:\quad
  \left\{\begin{aligned} 
   & \tilde y_j=y_j,\\
   & \tilde f_j=f_j,~~ \tilde\alpha_j=\alpha_j,\quad j\ne k,\\
   & \tilde f_k=f_k-\frac{\alpha_k}{y_k},~~ \tilde\alpha_k=-\alpha_k;
  \end{aligned}\right.\\
\label{yf.Bk}  
 B_k:\quad
  \left\{\begin{aligned}
   & \tilde f_j=-f_k-\frac{\lambda_j-\lambda_k}{f_j-f_k},~~ j\ne k,\quad \tilde f_k=-f_k,\\
   & \tilde y_j=\frac{f_j-f_k}{\lambda_j-\lambda_k}((f_j-f_k)y_j-\alpha_j),\quad 
     \tilde\alpha_j=\alpha_j,\quad j\ne k,\\
   & \tilde y_k=-y_k+6t(f^2_k+\lambda_k)-x-\sum_{j\ne k}(\tilde y_j+y_j),\quad 
     \tilde\alpha_k=1-\alpha_k.
  \end{aligned}\right.
\end{gather}
The Propositions \ref{prop:xt}, \ref{prop:xtLax} and \ref{prop:Bk} remain true also for these new variables. 

\section{Equations at $t=0$}\label{s:t0}

For the system (\ref{yft}) (or (\ref{yt})), the point $t=0$ is singular, due to the division by $t$ in the constraint (\ref{uy}). This means that a generic simultaneous solution of equations (\ref{yfx}) and (\ref{yft}) does not continue through the line $t=0$. However, there exist special solutions that are regular on this line. They are obtained if the initial condition at $t=0$ satisfies the system (\ref{yfx}) with the replacement of the expression (\ref{uy}) for $u$ by the constraints
\begin{equation}\label{y0}
 y_1+\dots+y_n=-\frac{x}{2},\quad
 2y_1f_1+\dots+2y_nf_n-\alpha_1-\dots-\alpha_n=-\frac{1}{2}.
\end{equation}
From here, in principle, one pair of variables can be eliminated, say, $y_n$ and $f_n$, which gives a closed system of equations for the remaining variables. The resulting formulas are rather cumbersome, but there exists a simple substitution that leads to the following symmetric system for a new set of variables $y_1,z_1,\dots,y_{n-1},z_{n-1}$:
\begin{gather}
\label{yzx}
 y_{j,x}=2y_j(z_j+v)-\alpha_j,\quad z_{j,x}=-z_j(z_j+2v)-\mu_j,\\
\label{v}
 v:=\frac{2}{x}(y_1z_1+\dots+y_{n-1}z_{n-1}+\beta).
\end{gather}
In the simplest case $n=2$ we have a second-order system
\[
 y_x= 2y\Bigl(z+\frac{2}{x}(yz+\beta)\Bigr)-\alpha,\quad 
 z_x= -z\Bigl(z+\frac{4}{x}(yz+\beta)\Bigr)-\mu.
\]
It is easy to check that if we solve the second equation for $y$ and substitute it into the first one, then for $z$ we get the P$_3$ equation 
\[
 z_{xx}=\frac{z^2_x}{z}-\frac{z_x}{x}+\frac{az^2+b}{x}+cz^3+\frac{d}{z}
\]
with the parameters values $a=4\beta+4\alpha-1$, $b=-(4\beta+1)\mu$, $c=1$, $d=\mu^2$.

\begin{proposition}\label{prop:yz}
The system (\ref{yfx}) with the constraints (\ref{y0}) is reduced to the form (\ref{yzx}), (\ref{v}) by the change
\[
 v=f_n,\quad z_j=f_j-f_n,\quad \mu_j=\lambda_j-\lambda_n,\quad 2\beta=\frac{1}{2}-\alpha_1-\dots-\alpha_n,
\]
moreover, the potential $u$ is reconstructed by the formula $u=v_x+v^2+\lambda_n$.
\end{proposition}

Rewriting the formulae (\ref{yf.Ak}) and (\ref{yf.Bk}) in new variables, we get the B\"acklund transformations for the system (\ref{yzx}), (\ref{v}). The parameters $\mu_j$ do not change for all transformations. Transformations $A_k$ and $B_k$ for $k=1,\dots,n-1$ have the following form:
\begin{equation}\label{yz.Ak}
 A_k:\quad
  \left\{\begin{aligned} 
   & \tilde y_j=y_j,\\
   & \tilde z_j=z_j,~~ \tilde\alpha_j=\alpha_j,\quad j\ne k,\\
   & \tilde z_k=z_k-\frac{\alpha_k}{y_k},~~ \tilde\alpha_k=-\alpha_k,~~\tilde\beta=\beta+\alpha_k,
  \end{aligned}\right.
\end{equation}
\begin{equation}\label{yz.Bk}  
 B_k:\quad
  \left\{\begin{aligned}
   & \tilde z_j=\frac{\mu_k}{z_k}-\frac{\mu_j-\mu_k}{z_j-z_k},~~ j\ne k,\quad \tilde z_k=\frac{\mu_k}{z_k},\\
   & \tilde y_j=\frac{z_j-z_k}{\mu_j-\mu_k}((z_j-z_k)y_j-\alpha_j),\quad j\ne k,\\
   & \tilde y_k=-\frac{x}{2}-\sum^{n-1}_{\substack{j=1\\ j\ne k}}\tilde y_j+\frac{z_k}{\mu_k}(z_ky_n+\alpha_n),\\
   & \tilde\alpha_j=\alpha_j,\quad j\ne k,\quad 
     \tilde\alpha_k=1-\alpha_k,\quad
     \tilde\beta=\beta+\alpha_k-\frac{1}{2}
  \end{aligned}\right.
\end{equation}
where $y_n$ and $\alpha_n$ denote the ``extra'' variable and parameter defined by relations
\[
 y_n:=-\frac{x}{2}-y_1-\dots-y_{n-1},\quad \alpha_n:=\frac{1}{2}-2\beta-\alpha_1-\dots-\alpha_n.
\]
Since the index $n$ is now distinguished, the formulae for $A_n$ and $B_n$ look different:
\begin{equation}\label{yz.An}
 A_n:\quad
 \tilde y_j=y_j,\quad \tilde z_j=z_j+\frac{\tilde\beta-\beta}{y_n},\quad 
 \tilde\alpha_j=\alpha_j,\quad \tilde\beta=\beta+\alpha_n,
\end{equation}
\begin{equation}\label{yz.Bn}  
 B_n:\quad
 \tilde z_j=-\frac{\mu_j}{z_j},\quad \tilde y_j=\frac{z_j}{\mu_j}(z_jy_j-\alpha_j),\quad
 \tilde\alpha_j=\alpha_j,\quad \tilde\beta=\beta+\alpha_n-\frac{1}{2}.
\end{equation}
However, in fact, these transformations are on equal footing with the rest ones, since instead of $y_n$ and $f_n$, any other pair of variables could be eliminated. The group properties of the transformations do not change, and they still generate the group ${\mathbb Z}^n_2\times{\mathbb Z}^n$, despite the decrease in the dimension of the system.

\begin{proposition}\label{prop:Bk'}
The transformations (\ref{yz.Ak})--(\ref{yz.Bn}) preserve the form of the systems (\ref{yzx}), (\ref{v}) and satisfy the group identities (\ref{ABgroup}).
\end{proposition}

\subsection*{Acknowledgements}

The work was done at Ufa Institute of Mathematics with the support by the grant \#21-11-00006 of the Russian Science Foundation, https://rscf.ru/project/21-11-00006/.



\begin{thebibliography}{99}

\bibitem{Adler_2020} V.E. Adler. Nonautonomous symmetries of the KdV equation and step-like solutions. 
 {\em J. Nonl. Math. Phys. \bf 27:3} (2020) \href{https://doi.org/10.1080/14029251.2020.1757236}{478--493}. 

\bibitem{Antonowicz_1992} M. Antonowicz. Gel'fand--Dikii hierarchies with sources and Lax representation for restricted flows.
 {\em Physics Letters A \bf 165:1} (1992) \href{https://doi.org/10.1016/0375-9601(92)91052-S}{47--52}.

\bibitem{Antonowicz_Rauch-Wojciechowski_1990} M. Antonowicz, S. Rauch-Wojciechowski. 
Constrained flows of integrable PDEs and bi-Hamiltonian structure of the Garnier system.
{\em Phys. Lett. A \bf 147:8--9} (1990) \href{https://doi.org/10.1016/0375-9601(90)90606-O}{455--462}.

\bibitem{Antonowicz_Rauch-Wojciechowski_1991} M. Antonowicz, S. Rauch-Wojciechowski. 
 Restricted flows of soliton hierarchies: coupled KdV and Harry Dym case.
 {\em J. Phys. A: Math. Gen. \bf 24:21} (1991) \href{https://doi.org/10.1088/0305-4470/24/21/017}{5043--5061}.

\bibitem{Burtsev_Zakharov_Mikhailov_1987} S.P. Burtsev, V.E. Zakharov and A.V. Mikhailov. Inverse scattering method with the variable spectral parameter. {\em Theor. Math. Phys. \bf 70:3} (1987) \href{http://dx.doi.org/10.1007/BF01040999}{227--240}.

\bibitem{Choodnovsky_1978} D.V. Choodnovsky, G.V. Choodnovsky. Completely integrable class of mechanical systems connected with Korteweg--de Vries and multicomponent Schr\"odinger equations -- I. {\em S\'eminaire sur les \'equations non lin\'eaires (Polytechnique)} (1977--1978), exp. no 6, p. \href{http://www.numdam.org/item/?id=SENL_1977-1978____A7_0}{1--9}.

\bibitem{Clarkson_Joshi_Pickering_1999} P.A. Clarkson, N. Joshi, A. Pickering. B\"acklund transformations for the second Painlev\'e hierarchy: a modified truncation approach. {\em Inverse Problems \bf 15} (1999) \href{https://doi.org/10.1088/0266-5611/15/1/019}{175--187}.

\bibitem{Dubrovin_2006} B. Dubrovin. On Hamiltonian perturbations of hyperbolic systems of conservation laws, II:  Universality of critical behaviour. {\em Commun. Math. Phys. \bf 267} (2006) \href{https://doi.org/10.1007/s00220-006-0021-5}{117--139}.

\bibitem{Fuchssteiner_1983} B. Fuchssteiner. Master symmetries, higher order time-dependent symmetries and conserved densities of nonlinear evolution equations. {\em Progr. Theor. Phys. \bf 70:6} (1983) \href{http://dx.doi.org/10.1143/PTP.70.1508}{1508--1522}.

\bibitem{Garnier_1919} R. Garnier. Sur une classe de syst\`emes diff\'erentielles Ab\'eliens d\'eduits de la th\'eorie
des \'equations lin\'eaires. {\em Rend. Circ. Mathem. Palermo \bf 43:4} (1919) \href{https://doi.org/10.1007/BF03014668}{155--191}.
 
\bibitem{GP_1973a} A.V.	Gurevich, L.P. Pitaevskii. Decay of initial discontinuity in the Korteweg--de Vries equation. {\em JETP Lett. \bf 17:5} (1973) \href{http://www.jetpletters.ru/ps/1540/article_23557.pdf}{193--195}.

\bibitem{GP_1973b} A.V.	Gurevich, L.P. Pitaevskii. Nonstationary structure of a collisionless shock wave. {\em Sov. Phys.-JETP \bf 38:2} (1974) \href{http://jetp.ras.ru/cgi-bin/dn/e_038_02_0291.pdf}{291--297}.

\bibitem{Hone_2001} A.N.W. Hone. Coupled Painlev\'e systems and quartic potentials. {\em J. Phys. A: Math. Gen. \bf 34:11} (2001) \href{https://doi.org/10.1088/0305-4470/34/11/316}{2235--2245}.

\bibitem{Ibragimov_Shabat_1979} N.H. Ibragimov, A.B.Shabat. Korteweg--de Vries equation from the group standpoint. {\em Dokl. Akad. Nauk SSSR \bf 244:1} (1979) \href{http://mi.mathnet.ru/eng/dan42245}{57--61}. [in Russian]

\bibitem{Kudashev_Suleimanov_1996} V. Kudashev, B. Suleimanov. A soft mechanism for generation the dissipationless shock waves. {\em Phys. Lett. A \bf 221:3} (1996) \href{https://doi.org/10.1016/0375-9601(96)00570-1}{204--208}. 

\bibitem{Kudryashov_1997} N.A. Kudryashov. The first and second Painlev\'e equations of higher order and some relations between them. {\em Phys. Lett. A \bf 224:6} (1997) \href{https://doi.org/10.1016/S0375-9601(96)00795-5}{353--360}.

\bibitem{Kudryashov_2002} N.A. Kudryashov. One generalization of the second Painlev\'e hierarchy.
{\em J. Phys. A: Math. Gen. \bf 35} (2002) \href{https://doi.org/10.1088/0305-4470/35/1/308}{93--99}.
 
\bibitem{Mazzocco_Mo_2007} M. Mazzocco, M.Y. Mo. The Hamiltonian structure of the second Painlev\'e hierarchy. {\em Nonlinearity \bf 20:12} (2007) \href{https://doi.org/10.1088/0951-7715/20/12/006}{2845--2882}.
 
\bibitem{Moore_1990} G. Moore. Geometry of the string equations. {\em Comm. Math. Phys. \bf 133:2} (1990) \href{https://doi.org/10.1007/BF02097368}{261--304}.

\bibitem{Novikov_1990} S.P. Novikov. Quantization of finite-gap potentials and nonlinear quasiclassical approximation in nonperturbative string theory. {\em Funct. Anal. Appl. \bf 24:4} (1990) \href{https://doi.org/10.1007/BF01077334}{296--306}.

\bibitem{Orlov_Shulman_1985} A.Yu. Orlov, E.I. Shulman. On additional symmetries of nonlinear Schr\"odinger equation. {\em Theor. Math. Phys. \bf 64:2} (1985) \href{http://dx.doi.org/10.1007/BF01017968}{862--866}.

\bibitem{Orlov_Rauch-Wojciechowski_1993} A.Yu. Orlov, S. Rauch-Wojciechowski. Dressing method, Darboux transformation and generalized restricted flows for the KdV hierarchy. {\em Physica D: Nonlinear Phenomena \bf 69:1--2} (1993)
\href{https://doi.org/10.1016/0167-2789(93)90181-Y}{77--84}.
 
\bibitem{Suleimanov_1994} B.I. Suleimanov. Onset of nondissipative shock waves and the ``nonperturbative'' quantum theory of gravitation. {\em JETP \bfseries 78:5} (1994) \href{http://jetp.ras.ru/cgi-bin/dn/e_078_05_0583.pdf}{583--587}.

\bibitem{Suris_2003} Yu.B. Suris. The problem of integrable discretization: Hamiltonian approach. Basel: Birkh\"auser, 2003.
 \url{https://doi.org/10.1007/978-3-0348-8016-9}

\bibitem{Veselov_Shabat_1993} A.P. Veselov, A.B. Shabat. Dressing chains and the spectral theory of the Schr\"odinger operators.
 {\em Funct. Anal. Appl. \bf 27:2} (1993) \href{https://doi.org/10.1007/BF01085979}{81--96}.

\end{thebibliography}
\end{document}